%% file: EmpiricalBayesSystemsMedicine.tex
\documentclass[a4paper,english]{article}
\usepackage{zibtitlepage}

\input{Packages}
\input{MacroFile1}

\begin{document}
\ZTPAuthor{Ilja Klebanov, Alexander Sikorski, Christof Sch\"utte, Susanna R\"oblitz}
\ZTPTitle{Empirical Bayes Methods for Prior Estimation in Systems
Medicine
\footnote{This research was carried out in the framework of {\sc Matheon} supported by Einstein Foundation Berlin, ECMath Project CH6.}
}
%\ZTPInfo{Best preprint of the year 2000}
\ZTPNumber{16-57}
\ZTPMonth{November}
\ZTPYear{2016}

\title{Empirical Bayes Methods for Prior Estimation in Systems Medicine}
\author{Ilja Klebanov}
\author{Alexander Sikorski}
\author{Christof Sch\"utte}
\author{Susanna R\"oblitz}
\zibtitlepage
\maketitle

\begin{abstract}
One of the main goals of mathematical modeling in systems medicine related to
medical applications is to obtain patient-specific parameterizations and model
predictions.
In clinical practice, however, the number of available measurements for single
patients is usually limited due to time and cost restrictions. This hampers the
process of making patient-specific predictions about the outcome of a
treatment. On the other hand, data are often available for many patients, in
particular if extensive clinical studies have been performed.
% \begin{rednote}
% Using these
% population data, we propose an iterative algorithm for constructing an
% informative prior distribution, which then serves as the basis for computing
% patient-specific posteriors and obtaining individual predictions. We demonstrate the performance of our method by applying it to a low-dimensional
% parameter estimation problem in a toy model as well as to a high-dimensional
% ODE model of the human menstrual cycle, which represents a typical example from
% systems biology modeling.
% \end{rednote}
Therefore, before applying Bayes' rule \emph{separately} to the data of each
patient (which is typically performed using a non-informative prior), it is meaningful
to use empirical Bayes methods in order to construct an informative prior from
all available data.
% Since non-parametric maximum likelihood estimation (NPMLE)
% for the prior leads to a discrete distribution with at most $M$ nodes, $M$ being
% the number of patients, a penalty term will be added to the marginal
% log-likelihood in order to avoid this ``peaked" behavior and favor ``smooth''
% priors.
We compare the performance of four priors -- a non-informative prior and priors
chosen by nonparametric maximum likelihood estimation (NPMLE), by maximum
penalized likelihood estimation (MPLE) and by doubly-smoothed maximum
likelihood estimation (DS-MLE) -- by applying them to a low-dimensional
parameter estimation problem in a toy model as well as to a high-dimensional
ODE model of the human menstrual cycle, which represents a typical example from
systems biology modeling.
\end{abstract}

\vspace{0.3cm}

\input{Introduction}
\input{Setup}

\input{Theory}
\input{ResultingAlgorithm}

\input{ToyExample}

\input{LargeExample}
\input{Conclusion}

\begin{appendix}
\input{Deconvolution}

\end{appendix}

\bibliographystyle{plain} % wmaainf, 3 letters and no blank
\bibliography{myBibliography}

\end{document}

%% file: Packages.tex
\usepackage{amsfonts,amssymb,amsmath,amsthm,amstext,amssymb,mathtools}
\usepackage{graphicx}
\usepackage{enumerate}
\usepackage{dsfont}  % f�r Doppel-Eins (Indikatorfunktion)
\usepackage{caption}
\usepackage{subcaption}
\usepackage{float}
\usepackage{color}

\usepackage{array}
\newcolumntype{C}[1]{>{\centering\arraybackslash}m{#1}}

\usepackage{tikz}
\usetikzlibrary{%
er,%
trees,%
shapes.symbols,%
mindmap,%
arrows,%
snakes,%
shapes.misc,%
shapes.arrows,%
chains,%
matrix,%
positioning,%
scopes,%
decorations.pathmorphing,%
shadows,%
decorations.pathmorphing,%
patterns%
}

\pgfdeclarelayer{bg}    % declare background layer
\pgfsetlayers{bg,main}  % set the order of the layers (main is the standard layer)

%% file: MacroFile1.tex
\usepackage{mathabx} % f�r := als \coloneq, stattdessen:
\mathchardef\ordinarycolon\mathcode`\:
\mathcode`\:=\string"8000
\begingroup \catcode`\:=\active
\gdef:{\mathrel{\mathop\ordinarycolon}}
\endgroup

\newcommand*{\Cdot}[1][2.4]{%
  \mathpalette{\CdotAux{#1}}\cdot%
}
\newdimen\CdotAxis
\newcommand*{\CdotAux}[3]{%
  {%
    \settoheight\CdotAxis{$#2\vcenter{}$}%
    \sbox0{%
      \raisebox\CdotAxis{%
        \scalebox{#1}{%
          \raisebox{-\CdotAxis}{%
            $\mathsurround=0pt #2#3$%
          }%
        }%
      }%
    }%
    \dp0=0pt %
    \sbox2{$#2\bullet$}%
    \ifdim\ht2<\ht0 %
      \ht0=\ht2 %
    \fi
    \sbox2{$\mathsurround=0pt #2#3$}%
    \hbox to \wd2{\hss\usebox{0}\hss}%
  }%
}
\definecolor{dunkelgruen}{rgb}{0,0.4,0}

\DeclareMathOperator*{\argmax}{arg\,max}

\def\R{\mathbb{R}}

\def\cX{\mathcal{X}}

\def\cZ{\mathcal{Z}}
\def\cM{\mathcal{M}}
\def\cN{\mathcal{N}}
\def\cW{\mathcal{W}}

\def\cL{\mathcal{L}}
\def\cI{\mathcal{I}}

\def\MPLE{\mathrm{MPLE}}

\def\NPMLE{\mathrm{NPMLE}}
\def\DSMLE{\text{DS-MLE}}
\def\true{\text{true}}

\def\impl{\Longrightarrow}

\theoremstyle{plain}
\newtheorem{theorem}{Theorem}

\newtheorem{proposition}[theorem]{Proposition}

\newtheorem{remark}[theorem]{Remark}

\theoremstyle{definition}

%% file: Introduction.tex
\section{Introduction}
\label{section:Introduction}
The estimation of a parameter $X\in\R^d$ from a measurement $Z\in\R^n$ given a
model $\phi\colon \R^d\to\R^n$ can be an ill-posed problem, especially if $\phi$
is non-injective. In addition, measurements are usually defective and the
measurement error should be taken into account. For these reasons, it is
meaningful to view $X$, $Z$ and the error term $E$ as random variables with
\begin{equation}
\label{equ:fehlerAdditiv}
Z = \phi(X) + E,
\qquad
E\sim\rho_E\ \ \text{independent of $X$},
\end{equation}
where $\rho_E$ denotes the density of the error term $E$. After establishing a
\emph{prior density} $\rho_X(x)$ of $X$, which reflects our initial knowledge
about the parameter, Bayes' rule gives us the proper tool to update our
knowledge when taking into accout the measurement:
\begin{align*}
&\rho_X(x\mid Z=z) = \frac{\rho_E(z-\phi(x))\, \rho_X(x)}{\rho_Z(z)}\, ,
\intertext{where}
&\rho_Z(z)
=
\int_{\R^d}
\rho_E(z-\phi(x))\, \rho_X(x)\, \mathrm dx.
\end{align*}
In contrast to point estimators like the maximum likelihood estimator (MLE)
\[
\hat x_{\rm MLE} = \argmax_x \rho_E(z-\phi(x)),
\]
which usually try to optimize the fit, the result of the Bayesian inference is a
whole \emph{distribution}, the "posterior", in the parameter space $\cX$.
Usually, especially in high dimensions, it is given by a (possibly weighted)
sampling, since the density function is too costly to compute due to the
integral $\rho_Z(z)$ in the denominator. A sampling can be produced without
the knowledge of this denominator by e.g. the Metropolis-Hastings algorithm
\cite{metropolis1953equation,hastings1970monte} and is useful for the
computation of certain expectation values with respect to the posterior.

However, in many applications no comprehensible prior can be assigned, which
results in eliciting priors based on expert opinion and therefore in
different posterior distributions depending on which expert was asked.
This unsatisfactory lack of objectivity and consistency has led to many
controversial discussions about the reasonability and trustability of Bayesian
inference.

Empirical Bayes methods provide one possible solution to this issue by first
estimating the prior distribution. This, naturally, will require further
knowledge, in our case (independent, $\rho_Z$-distributed) measurements
$Z_1=z_1,\dots,Z_M=z_M$ for \emph{several} individuals, which exists
in many statistical trials.

An important application is the prediction of patient-specific treatment
success rates based on clinical measurement data and a mathematical model
describing the underlying physiological processes. One example are hormonal
treatments of the human menstrual cycle, as they are frequently performed in
reproductive medicine. In this case, clinical data are available as well as a
robust mathematical model, which allows to simulate the cyclic behavior under
varying external conditions \cite{Roeblitz2013}.
 
Typically, predictions are required for a specific patient in the daily clinical
practice, where the number of measurements is limited due to time and cost
restrictions.  On the other hand, data are often available for many (hundreds or
even thousands of) patients, in particular if extensive clinical studies have
been performed for the approval of a drug.
Using these population data, we propose an iterative algorithm for constructing
an informative prior distribution, which then serves as the basis for computing
patient-specific posteriors and obtaining individual predictions.

In the empirical Bayes framework, the prior $\Pi=\pi$ is considered as a
hyperparameter with unknown true value $\pi_{\rm true}=\rho_X$. A corresponding
likelihood function $L(\pi)$ can be derived and statistical inference
(frequentist or Bayesian) can be applied for its estimation.
If the prior has no prescribed parametric form, the resulting nonparametric
maximum likelihood estimate (NPMLE)
\begin{equation}
\label{equ:MPLEformulation}
\pi^\NPMLE = \argmax_\pi \ \log L(\pi)
\end{equation}
is given by a discrete distribution
with at most $M$ nodes, see \cite[Theorems 2-5]{laird1978nonparametric} or
\cite[Theorem 21]{lindsay1995mixture}.
In order to avoid this kind of ``peaked'' behavior (overfitting of the MLE) and
to favor ``smooth'' priors, a penalty term $\Phi(\pi)$ is often subtracted from
the marginal log-likelihood. The resulting maximum penalized likelihood
estimate (MPLE),
\[
\pi^\MPLE = \argmax_\pi \ \log L(\pi) - \Phi(\pi),
\]
is a trade-off between goodness of fit and smoothness of the prior.

Following the discussion in \cite{theoretischesPaper}, we will use the negative mutual information
$\cI[X;Z](\pi)$ as our penalty, which is equivalent to using the negative
``$Z$-entropy'' $H_Z(\pi)$ (the entropy in measurement space) in the case of an
additive error \eqref{equ:fehlerAdditiv},
\begin{align}
\begin{split}
\label{equ:Zentropy}
\Phi(\pi)
&=
- \gamma\, H_Z(\pi) := \gamma \int
\rho_Z(z|\Pi=\pi)\, \log \rho_Z(z|\Pi=\pi)\, \mathrm dz
\\
&=
- \gamma\, \cI[X;Z](\pi) + {\rm const.},
\end{split}
\end{align}
where the weight $\gamma$ equilibrates the trade-off between smoothness of the
prior and goodness of fit. The resulting prior estimate $\pi^\MPLE$ has a
beautiful connection to reference priors \cite{berger1992development} -- the two coincide in the case of no measurements. As an alternative
approach, we will apply the doubly-smoothed MLE (DS-MLE) introduced by Seo and
Lindsay in \cite{seo2013universally,seo2010computational} to our problem, which
can be viewed as a regularization in the measurement space before the
application of NPMLE.

This paper has a theoretical counterpart \cite{theoretischesPaper}, where the
theory is explained in detail. Here, we will concentrate on applications.

We will introduce the notation in Section \ref{section:Setup}, set out the
theory and derive the numerical scheme in Section \ref{section:Theory}. The
resulting algorithm for two different scenarios and their implementation are
discussed in Section \ref{section:ResultingAlgorithm}.
Corresponding numerical results for a low-dimensional toy example are presented
in Section \ref{section:ToyExample}, together with an example that shows what
can go wrong.
Finally, in Section \ref{section:LargeExample} the algorithm is applied to a high-dimensional
parameter estimation problem in an ODE model of the human menstrual cycle, which
represents a typical example from systems medicine modeling.

%% file: Setup.tex
\section{Setup and Notation}
\label{section:Setup}
Throughout this manuscript, we will use the following notation:
\begin{enumerate}[(i)]
 \item The probability density function of a random variable $Y$ will be
 denoted by $\rho_Y$, while $\rho_Y(\Cdot\mid A)$ will stand for its conditional
 density given an event $A$ (typically $A = \{X=x\}$ or $A = \{\Pi=\pi\}$).
 Other probabilities will be denoted by $\pi$ for ``priors'' and $p$ for
 ``posteriors'', in particular,
 \[
 p_\pi^z(x) :=  \frac{\rho_Z(z\mid X=x)\, \pi(x)}{\int \rho_Z(z\mid X=\tilde
 x)\, \pi(\tilde x)\, \mathrm d\tilde x}
 \]
 denotes the posterior density of $X$ given the measurement $Z=z$ and the prior
 $\pi$.
 \item In the case of an additive error \eqref{equ:fehlerAdditiv}, which we
 will always assume, the likelihood model $\{\rho_Z(\Cdot\mid X=x)\mid
 x\in\R^d\}$ is given by
 \[
 \rho_Z(z\mid X=x) := \rho_E\left(z-\phi(x)\right),
 \]
 where $\phi\colon\R^d\to\R^n$ is the (known) underlying model and $\rho_E$ is
 the (known) probability density of the additive error term $E$.
 \item
 Since we assume to have several patients with (independent and 
 $\rho_X$-distributed, but unknown) parametrizations $X_m$ and
 (known) measurements $Z_m$,
 \[
 X_m\stackrel{\rm i.i.d.}{\sim}\pi_{\rm true} = \rho_X,
 \qquad
 Z_m \stackrel{\rm indep.}{\sim}\rho_Z(\Cdot\mid X=x_m),
 \qquad
 m=1,\dots,M,
 \]
 our likelihood model $\{\rho_Z(\Cdot\mid\Pi=\pi)\mid \pi\in\cM_1(\R^d)\}$ for
 the hyperparameter $\Pi=\pi$, where $\cM_1(\R^d)$ denotes the set of all
 probability densities on $\R^d$, is given by
 \begin{equation}
 \label{equ:hyperLikelihood}
 \rho_Z(\Cdot\mid\Pi=\pi) = \int \rho_Z(z\mid X=x)\, \pi(x)\, \mathrm dx,
 \end{equation}
  which is the ``would-be probability density'' of $Z$, if $\pi$ was the true
  prior (so, the true density of $Z$ is given by $\rho_Z = \rho_Z(\Cdot\mid
 \Pi=\pi_{\rm true})$).
 The (marginal) likelihood $L(\pi)$ is then given by
 \[
 L(\pi) = \prod_{m=1}^M \rho_Z(z_m\mid\Pi=\pi).
 \]
 We will call the likelihood model identifiable (see e.g.
 \cite[Section 5.5]{van2000asymptotic}), if
 \begin{equation}
 \label{equ:hyperIdentifiability}
 \rho_Z(\Cdot\mid\Pi=\pi) = \rho_Z(\Cdot\mid\Pi=\pi_{\rm true})
 \ \Longleftrightarrow\
 \pi=\pi_{\rm true}.
 \end{equation} 
\end{enumerate}
We will slightly abuse notation by using \eqref{equ:hyperLikelihood} in the
more general case $\pi\in L^1(\R^d)$ and by utilizing the same notation for
probability densities and the corresponding probability distributions.

%% file: Theory.tex
\section{Theory}
\label{section:Theory}
In order to reproduce the probability density $\rho_X$ of the parameters
from measurements $z_1,\dots,z_M$, we will recursively apply an approximation
to the fixed point iteration
\vspace{0.4cm}
\begin{center}
\fbox{\parbox{0.85\textwidth}{
\vspace{-0.3cm}
\begin{align}
\label{equ:fixedPointIteration}
\pi_{n+1}(x)
= &
\left(\Psi \pi_n\right)(x),\ \text{ where}
\\
\label{equ:IterationPsi}
\left(\Psi \pi\right)(x)
:= &
\int p_{\pi}^z(x)\, \rho_Z(z)\, \mathrm dz
=
\pi(x) \int \frac{\rho_Z(z\, |\, X=x)}{\rho_{Z}(z\mid\Pi=\pi)}\, \rho_Z(z)\,
\mathrm dz.
\end{align}
\vspace{-0.3cm}
}
}
\end{center}
\vspace{0.4cm}
This iteration is motivated by the observation that the true prior density
$\pi_{\rm true}=\rho_X$ of $X$ is a fixed point of $\Psi$,
\begin{equation}
\label{equ:fixedPoint}
\left(\Psi\pi_{\rm true}\right)(x)
=
\pi_{\rm true}(x) \int \frac{\rho_Z(z\, |\, X=x)}{\rho_{Z}(z)}\, \rho_Z(z)\,
\mathrm dz =
\pi_{\rm true}(x),
\end{equation}
and can be seen as an analogue to the EM algorithm applied to the ``infinite
data'' log-likelihood,
\begin{equation}
\label{equ:ccLogLikelihood}
\cL_{\rm cc}(\pi):=
-H^{\rm cross}\big( \rho_Z,\, \rho_Z(z\mid\Pi=\pi) \big)
=
\int \rho_Z\, \log \rho_Z(z\mid\Pi=\pi)\, \mathrm dz,
\end{equation}
where $H^{\rm cross}$ denotes the cross entropy.
It can be shown that $\cL_{\rm cc}$ is concave in $\pi$ and that
its value is increased in each iteration step (see \cite{theoretischesPaper} ):
\[
\cL_{\rm cc}(\Psi\pi)\ge\cL_{\rm cc}(\pi)\quad \text{for all }\pi\in\cM_1(\R^d).
\]
More precisely, the following statement holds:

\begin{proposition}
\label{prop:equivalenceGoodPi}
Let $\pi\in \cM_1(\R^d)$ be a globally supported probability density
function. Then the following two statements are equivalent:
\begin{enumerate}[(i)]
  \item $\displaystyle \rho_{Z}(\Cdot \mid\Pi=\pi) = \rho_Z$,
  \item $\displaystyle \Psi\pi=\pi$,
  \item $\pi$ maximizes $\cL_{\rm cc}(\pi)$.
\end{enumerate}
\end{proposition}
\begin{proof}
The proof for (i)$\Rightarrow$(ii) goes analogously to \eqref{equ:fixedPoint}.
For (ii)$\Rightarrow$(i) we will use the abbreviation
\[
\rho_{Z,f} := \rho_Z(\Cdot\mid\Pi=f),\qquad f\in L^1(\R^d).
\]
We define the subspace
\[
\mathcal E = \left\{\rho_{Z,f}\mid f\in L^1(\R^d)\right\}\subseteq L^1(\R^n)
\]
with weighted $L^2$ inner product
\[
\left\langle\rho_{Z,f_1}, \rho_{Z,f_2}\right\rangle_{\pi}
:=
\int_{\R^n}\frac{\rho_{Z,f_1}(z)\,
\rho_{Z,f_2}(z)}{\rho_{Z,\pi}(z)}\, \mathrm dz.
\]
We can formulate the following chain of implications:
\begin{align*}
\text{(ii)}
&\ \impl\
\forall x:\quad
\int \frac{\rho_Z(z)}{\rho_{Z,\pi}(z)}\, \rho_Z(z\, |\, X=x)\,
\mathrm dz = 1
\\
&\ \impl\
\forall x:\quad
\int \left(1-\frac{\rho_Z(z)}{\rho_{Z,\pi}(z)}\right) \rho_Z(z\, |\,
X=x)\, \mathrm dz = 0
\\
&\ \impl\
\int (\pi-\rho_X)(x)\int
\frac{\rho_{Z,\pi}(z)-\rho_Z(z)}{\rho_{Z,\pi}(z)}\, \rho_Z(z\,
|\, X=x)\, \mathrm dz\, \mathrm dx = 0
\\
&\ \impl\
\int
\frac{\rho_{Z,\pi}(z)-\rho_Z(z)}{\rho_{Z,\pi}(z)}\,
\rho_{Z,(\pi-\rho_X)}(z)\, \mathrm dz = 0
\\
&\ \impl\
\left\langle \rho_{Z,\pi} -\rho_Z,
\rho_{Z,\pi}-\rho_Z\right\rangle_{\pi} = 0,
\end{align*}
which implies (i) by the positive definiteness of the inner product.

The equivalence of (i) and (iii) is given by the fact that the cross entropy of
two densities is minimal if and only if the two densities agree.
\end{proof}

So, in the identifiable case \eqref{equ:hyperIdentifiability}, $\pi_{\rm true}$
is the only fixed point of the iteration \eqref{equ:fixedPointIteration} and one
can prove convergence. 

As discussed in \cite{theoretischesPaper},
Proposition \ref{prop:equivalenceGoodPi} suggests yet another estimation
approach for the density $\pi_{\rm true} = \rho_X$:
Compute an approximation $\rho_Z^{\rm appr}$ of the density $\rho_Z$ using the
measurements $Z_1,\dots,Z_M$ and then minimize the cross entropy between
$\rho_Z^{\rm appr}$ and $\rho_Z(z\mid\Pi=\pi)$:
\[
\pi_\ast = \argmax_\pi \cL_{\rm cc}^{\rm appr}(\pi),
\qquad
\cL_{\rm cc}^{\rm appr}(\pi)
:=
-H^{\rm cross}\big( \rho_Z^{\rm appr},\, \rho_Z(z\mid\Pi=\pi) \big).
\]
If the approximation of $\rho_Z$ is performed by kernel density estimation, the
resulting method is the so-called doubly-smoothed maximum likelihood estimation
(DS-MLE) introduced by Seo and Lindsay in
\cite{seo2013universally,seo2010computational}. Note that, in this case, due to
the additional smoothing by the kernel, the likelihood model has to be smoothed
as well in order to get consistent results. We will denote the resulting
density estimate $\pi_\ast$ by $\pi_\DSMLE$.

\subsection{Numerical realization}
\label{section:TheoryNumerical}
The numerical approximation of the fixed point iteration
\eqref{equ:fixedPointIteration}--\ref{equ:IterationPsi} will be realized by two
discretization steps:
\begin{enumerate}[(i)]
  \item The first discretization in $\R^n$ is due to the fact that we have
  only finitely many $\rho_Z$-distributed measurements $Z_1,\dots,Z_M$ instead
  of the density $\rho_Z$ appearing in \eqref{equ:IterationPsi}. We will
  use the following Monte-Carlo approximation:
  \begin{equation}
  \label{equ:MCinZ}
  \tag{Z-MC}
  \rho_Z\approx \frac{1}{M}\sum_{m=1}^M\delta_{z_m}.
  \end{equation}
  \item In order to compute the high-dimensional integrals
  $\rho_Z(z_m\mid\Pi=\pi)$, a second discretization in $\R^d$ is necessary,
  which will be realized by another Monte-Carlo approximation for prior densities
  $\pi$ (and the hyperparameter $\Pi=\pi$ will be replaced by $W=w$):
  \begin{small}
  \begin{equation}
  \label{equ:MCinX}
  \tag{X-MC}
  \pi\approx \sum_{k=1}^K w_k\delta_{x_k},
  \quad  
  x_k\in\R^d,\ w\in\cW :=
  \Big\{w\in\R^K\mid w_k\ge 0\, \forall k,\,
  \sum_{k=1}^K w_k = 1\Big\}.
  \end{equation}
  \end{small}
\end{enumerate}
The application of these discretizations to the infinite data log-likelihood
$\cL_{\rm cc}$ defined by \eqref{equ:ccLogLikelihood} and to the fixed point
iteration $\Psi_{\rm cc}:=\Psi$ defined by \eqref{equ:fixedPointIteration} --
\eqref{equ:IterationPsi} results in the commutative diagram displayed in Figure
\ref{fig:commutativeDiagram}, which is discussed in detail in \cite{theoretischesPaper}.

\begin{figure} %[h!]
\begin{center}
\begin{tikzpicture}[description/.style={fill=white,inner sep=2pt}]
\matrix (m) [matrix of math nodes, row sep=2.7em,
column sep=1.5em, text height=1.5ex, text depth=0.25ex]
{
& \boxed{\cL_{\rm cc}} &&&&& \boxed{\Psi_{\rm cc}} && \\
&&& \boxed{\cL_{\rm dc}} &&&&& \boxed{\Psi_{\rm dc}} \\
 \boxed{\cL_{\rm cd}} &&&&& \boxed{\Psi_{\rm cd}} &&& \\
&& \boxed{\cL_{\rm dd}} &&&&& \boxed{\Psi_{\rm dd}} & \\
};

\path[->,line width=1pt, draw=blue]
(m-1-7)	edge node [right] {\footnotesize $\ \ \textcolor{blue}{\text{(X-MC)}}$}
(m-2-9) edge node [above=0.4cm] {\footnotesize $\ \
\textcolor{blue}{\text{(Z-MC)}}\qquad$} (m-3-6)

(m-2-9)	edge node [right]
{\footnotesize $\textcolor{blue}{\text{(Z-MC)}}$} (m-4-8) (m-3-6)	edge node [above]
{\footnotesize $\ \ \textcolor{blue}{\text{(X-MC)}}$} (m-4-8)

(m-1-2)	edge node [right] {\footnotesize $\ \ \textcolor{blue}{\text{(X-MC)}}$}
(m-2-4) edge node [above left = -0.2cm] {\footnotesize
$\textcolor{blue}{\text{(Z-MC)}}$} (m-3-1) edge [line width=1pt,
draw=red] node [above] {\footnotesize\textcolor{red}{EM}}(m-1-7)
(m-3-1)	edge node [below left = -0.2cm] {\footnotesize
$\textcolor{blue}{\text{(X-MC)}}$} (m-4-3) edge [line width=1pt, draw=red] (m-3-6)
(m-2-4)	edge node [below = 0.4cm] {\footnotesize
$\textcolor{blue}{\quad\ \text{(Z-MC)}}$} (m-4-3) edge [line width=1pt,
draw=red] (m-2-9) (m-4-3)	edge [line width=1pt, draw=red] node [below] {\footnotesize $\textcolor{red}{\mathrm{EM}}$} (m-4-8) ;

\draw  (3,0.5) node {\footnotesize $\textcolor{red}{\mathrm{EM}}$};
\draw  (-3,-0.5) node {\footnotesize $\textcolor{red}{\mathrm{EM}}$};
% \draw  (1.8,1.5) node {\footnotesize $\textcolor{blue}{\text{(Z-MC)}}$};
% \draw  (-1.9,-1.5) node {\footnotesize $\textcolor{blue}{\text{(Z-MC)}}$};

\end{tikzpicture}
\caption{The relations between the log-likelihoods $\cL$ and the fixed point
iterations $\Psi$ resulting from the application of the EM algorithm summarized
by a commutative diagram.}
\label{fig:commutativeDiagram}
\end{center}
\end{figure}
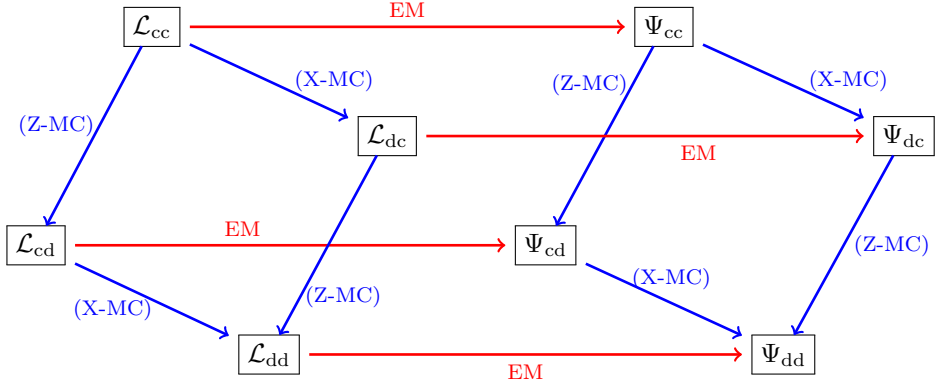

Here, the indices c and d denote whether the parameter space (first index) and
the measurement space (second index) are considered continuous or discretized in
the above sense. The corresponding log-likelihoods and fixed point iterations
are given by
\begin{small}
\begin{align*}
&
\cL_{\rm cc}(\pi)
=
\int_\cZ \rho_Z(z) \log\rho_{Z}(z|\Pi=\pi)\, \mathrm dz,
&&&
&\Psi_{\rm cc} \pi\,  (x)
=
\pi(x)\int_\cZ \rho_Z(z)\, \frac{\rho_Z(z|X=x)}{\rho_{Z}(z|\Pi=\pi)}\, \mathrm
dz,&
\\[0.1cm]
&
\cL_{\rm cd}(\pi)% := \frac{1}{M}\log L(\pi)
= \frac{1}{M}
\sum_{m=1}^M \log \rho_{Z}(z_m|\Pi=\pi),
&&&
&\Psi_{\rm cd} \pi\,  (x)
=
\frac{\pi(x)}{M}\sum_{m=1}^M
\frac{\rho_Z(z_m|X=x)}{\rho_{Z}(z_m|\Pi=\pi)}\, ,&
\\[0.1cm]
&
\cL_{\rm dc}(w)
=
\int_\cZ \rho_Z(z) \log\rho_{Z}(z|W=w)\, \mathrm dz,
&&&
&\left[\Psi_{\rm dc} w\right]_k
=
w_k\int_\cZ \rho_Z(z)\, \frac{\rho_Z(z|X=x_k)}{\rho_{Z}(z|W=w)}\, \mathrm dz,&
\\[0.1cm]
&
\cL_{\rm dd}(w)
=
\frac{1}{M}\sum_{m=1}^M \log \rho_{Z}(z_m|W=w),
&&&
&\left[\Psi_{\rm dd} w\right]_k
=
\frac{w_k}{M}\sum_{m=1}^M
\frac{\rho_Z(z_m|X=x_k)}{\rho_{Z}(z_m|W=w)}\, .&
\end{align*}
\end{small}

\begin{remark}
Note that $\cL_{\rm cd}$ and $\cL_{\rm dd}$ are the log-likelihood functions of
the hyperparameter $\Pi=\pi$, $W=w$ respectively (divided by $M$). The fixed
point iterations $\Psi_{\rm cd}$ and $\Psi_{\rm dd}$ resulting from the
application of the EM algorithm have been studied before, see e.g.
\cite{turnbull1976empirical,laird1978nonparametric}. They are often referred to
as the ``self-consistency algorithm'' since they fulfill the self-consistency
principle introduced by Efron \cite{efron1967two}.
Therefore, the stepwise increase of $\cL_{\rm cd}$ and $\cL_{\rm dd}$ when
applying the iterations $\Psi_{\rm cd}$ and $\Psi_{\rm dd}$, respectively,
follows from the theory on the EM algorithm \cite{dempster1977maximum}.
The proof of an analogous statement for $\Psi_{\rm cc}$ and $\Psi_{\rm dc}$ is
given in \cite{theoretischesPaper}.
\end{remark}

In order to guarantee that the nodes $x_1,\dots,x_K$ from the discretization
\eqref{equ:MCinX} lie in an area of high probability of the iterates
$\pi_n$, they will be chosen $\pi_1$-distributed, since our experiments
have shown that $\pi_1$ already gives a considerable improvement over $\pi_0$
in approximating $\pi_{\rm true}$.
This will be realized by producing a $p_{\pi_0}^{z_m}$-distributed sampling
$(x_1^{(m)},\dots,x_L^{(m)})$ for each $m=1,\dots,M$ by the Metropolis-Hastings
algorithm and merging these to get a $\pi_1$-distributed sampling
\[
\mathcal X = \{x_1,\dots,x_K\} = \bigcup_{m=1}^M
\left\{x_1^{(m)},\dots,x_L^{(m)}\right\}
\]
with weights $w_{1,k} = 1/K$, $k=1,\dots,K$.

The numerical realization of DS-MLE goes analogously to NPMLE but with an
augmented set of measurements $\zeta_1,\dots,\zeta_{\tilde M}$ instead of
$Z_1,\dots,Z_M$, which are samples from the kernel density estimate $\rho_Z^{\rm
appr}$, see the description in \cite{theoretischesPaper} or \cite{seo2010computational}.

Unfortunately, even though $\cL_{\rm cd}$ is an approximation of $\cL_{\rm cc}$
and $\cL_{\rm cc}$ is maximized by the true prior $\pi_{\rm true}$  of
$\cL_{\rm cc}$, $\cL_{\rm cd}$ can be shown to be maximized by a discrete
distribution with at most $M$ nodes, see \cite[Theorems 2-5]{laird1978nonparametric} or \cite[Theorem
21]{lindsay1995mixture}.
This originates from the typical ``overconfidence'' (or ``overfitting'') of maximum
likelihood estimators and results in a poor approximation of the true prior:
\[
\pi_{\rm cd,\infty} = \argmax_\pi\ \cL_{\rm cd}(\pi)\, \napprox\, \argmax_\pi\ 
\cL_{\rm cc}(\pi) = \pi_{\rm true}.
\]
As discussed in the introduction, one way to deal with this issue is to
subtract a penalty term from $\cL_{\rm cd}$ before maximizing it, which we will
choose to be the negative entropy $-H_Z(\pi)$ in the measurement space defined
by \eqref{equ:Zentropy}:
\begin{equation}
\label{equ:piMPLE}
\pi^\MPLE = \argmax_\pi \ M\cL_{\rm cd}(\pi) + \gamma H_Z(\pi).
\end{equation}
The numerical implementation of \eqref{equ:piMPLE} will be realized by a gradient
ascent of the discretized functional
\begin{align*}
\mathcal A
&:=
M\cL_{\rm dd}(w) + \gamma H_Z(w),
\intertext{where}
H_Z(w)
&:=
-\int \rho_Z(z\mid W=w)\, \log \rho_Z(z\mid W=w)\, \mathrm dz.
\end{align*}
% $\mathcal A := M\cL_{\rm dd}(w) + H_Z(w)$,
% where
% \[
% H_Z(w) := -\int \rho_Z(z\mid W=w)\, \log \rho_Z(z\mid W=w)\, \mathrm dz.
% \]
In order for the iterates $w_n$ to stay in the simplex $\cW$ as defined in
\eqref{equ:MCinX}, an additional bestapproximation is necessary in each step.

The integral appearing in the gradient of the functional $\mathcal A$,
\[
\nabla\mathcal A
=
\left(
\sum_{m=1}^M \frac{\rho_Z(z_m\mid X = x_k)}{\rho_Z(z_m\mid W=w)}
- \gamma\underbrace{\int\rho_Z(z\mid X=x_k)\, \log(\rho_Z(z\mid W=w))\, \mathrm
dz}_{=:\, I} -\gamma \right)_{k=1}^K,
\]
is computed using an importance sampling,
\begin{align*}
I
&=
\int\frac{\rho_Z(z\mid X=x_k)}{\rho_Z(z\mid W=w)}\, \log(\rho_Z(z\mid W=w))\,
\rho_Z(z\mid W=w)\, \mathrm dz
\\
&\approx
\sum_{j=1}^J \frac{\rho_Z(z_j\mid X=x_k)}{\rho_Z(z\mid W=w)}\,
\log(\rho_Z(z_j\mid W=w)),
\end{align*}
where the points $z_j$ are chosen $\rho_Z(\Cdot\mid W=w)$-distributed.

\subsection{Non-identifiable case}
\label{section:nonidentifiable}
If the identifiability assumption \eqref{equ:hyperIdentifiability} is not
fulfilled, we cannot expect the fixed point iteration
\eqref{equ:fixedPointIteration}--\eqref{equ:IterationPsi} to converge to the
true prior $\pi_{\rm true} = \rho_X$. The best we can hope for is to get
close to a prior $\pi$ such that $\rho_Z(\Cdot\mid\Pi=\pi) = \rho_Z$.
Therefore, one way to enforce convergence is to restrict ourselves to
equivalence classes of densities with respect to the equivalence relation
\[
\pi\sim\pi'
\ \Longleftrightarrow\
\left\| \rho_Z(\Cdot\mid\Pi=\pi) - \rho_Z(\Cdot\mid\Pi=\pi')\right\|_{L^1(\R^n)}=0.
\]
Note that the set  of equivalence classes $L^1(\R^d)/{\sim}$ is, in fact, the quotient space $L^1(\R^d)/\mathrm{ker}(\psi)$ emerging
from the linear map
\[
\psi\colon L^1(\R^d)\to L^1(\R^n),
\qquad
\pi\mapsto \rho_Z(\Cdot\mid\Pi=\pi) = \int \rho_Z(\Cdot\mid X=x)\,
\pi(x)\, \mathrm dx.
\]
Therefore, $L^1(\R^d)/{\sim}$ inherits the $L^1$-norm and $L^1$-distance via
\[
\left\|[\pi]\right\|_{L^1}
=
\inf_{\pi'\in [\pi]} \left\|\pi'\right\|_{L^1},
\qquad
\left\|[\pi_1] - [\pi_2]\right\|_{L^1}
=
\inf_{\pi_1'\in [\pi_1]\atop \pi_2'\in [\pi_2]}
\left\|\pi_1'-\pi_2'\right\|_{L^1}
\]
and we can choose from the following two definitions for the convergence of
$\pi_n$ to $\pi_{\infty}$:
\begin{align*}
&\pi_n\xrightarrow[Z]{n\to\infty}\pi_{\infty}
\ :\Longleftrightarrow\
\left\|[\pi_n] -
[\pi_\infty]\right\|_{L^1}\xrightarrow{n\to\infty}0
\intertext{or}
&\pi_n\xrightarrow[Z]{n\to\infty}\pi_{\infty}
\ :\Longleftrightarrow\
\left\|\rho_Z(\Cdot\mid\Pi=\pi_n) -
\rho_Z(\Cdot\mid\Pi=\pi_\infty)\right\|_{L^1}\xrightarrow{n\to\infty}0.
\end{align*}
Both definitions are meaningful but the second definition yields a weaker form
of convergence as stated by the following proposition:
\begin{proposition}
Let $\pi\in L^1(\R^d)$. Then, in the above notation,
\[
\left\|\rho_Z(\Cdot\mid\Pi=\pi)\right\|_{L^1}
\le
\left\|[\pi]\right\|_{L^1}\, .
\]
\end{proposition}
\begin{proof}
For each $\pi'\in[\pi]$ we have
\begin{align*}
\left\|\rho_Z(\Cdot\mid\Pi=\pi)\right\|_{L^1}
&\ =\
\big\|\rho_Z(\Cdot\mid\Pi=\pi')\big\|_{L^1}
\ =\
\int\left| \int \rho_Z(z\, |\, X=x)\, \pi'(x)\, \mathrm dx\, \right|
\mathrm dz
\\
&\ \le\
\int \left|\pi'(x)\right|
\underbrace{\int \rho_Z(z\, |\, X=x)\, \mathrm dz}_{=1}\, \mathrm dx
\ =\
\left\|\pi'\right\|_{L^1}.
\end{align*}
\end{proof}

%% file: ResultingAlgorithm.tex
\section{Resulting Algorithm}
\label{section:ResultingAlgorithm}
Given the data $\mathcal Z=\{z_1,\dots,z_M\}\subseteq \R^n$ of $M$ patients, we
will discuss two scenarios:
\begin{enumerate}[(A)]
  \item No diagnoses have been made.
  \item The patients have been diagnosed with diseases/sicknesses
  $s_1,\dots,s_L$ (for simplicity, we will assume that each patient has exactly
  one disease), resulting in a partition of the data set, where
  $Z^{(l)}$ denotes the data of the patients with disease $s_l$:
  \[
  \mathcal Z = \{z_1,\dots,z_M\} = \bigsqcup_{l=1}^L \mathcal Z^{(l)},
  \qquad
  \mathcal Z^{(l)}=\{z_1^{(l)},\dots,z_{M_l}^{(l)}\}
  \qquad
  (M = \sum_{l=1}^L M_l).
  \]   
\end{enumerate}

\textbf{Approach for scenario (A):}
\begin{itemize}
  \item
  Starting with a non-informative prior $\pi_0$, we construct informative priors
  $\pi_\NPMLE$, $\pi_\DSMLE$ and $\pi_\MPLE$ by the fixed point
  iterations discussed in Section \ref{section:TheoryNumerical}.    
  All four priors, $\pi = \pi_0$, $\pi_{\DSMLE}$, $\pi_{\NPMLE}$, $\pi_\MPLE$,
  will be given by the same sampling $\mathcal X = \{x_1,\dots,x_K\}$ but with
  weights $w = w_{0}$, $w_\NPMLE$, $w_\DSMLE$, $w_\MPLE$ and their performance
  will be compared in the next steps.
  \item
  Given a patient with data $Z = z_{\ast}$, we compute the individual
  posteriors $p_{\pi}^{z_\ast}$ with respect to these priors, which will be
  given by new (individual) weights,
  \[
  v_k^{\ast} = \frac{w_k\, \rho_Z(z^\ast\mid X = x_k)}{\sum_{j=1}^K
  w_j\, \rho_Z(z^\ast\mid X=x_j)}\, .
  \]
  \item
  The success rate $R_\ast$ can now be approximated via
   \begin{align*}
	R_\ast
&=
\int r(x)\, p_{\pi}^{z_\ast}(x)\, \mathrm dx
\approx
\frac{1}{K}\sum_{k=1}^K v_k^{\ast}\, r(x_k),
\intertext{where}
r(x)
&=
\begin{cases}
1 &\text{if the treatment, given the parameters $x$, is successful,}\\
0 &\text{otherwise.}
\end{cases}
\end{align*}
\end{itemize}

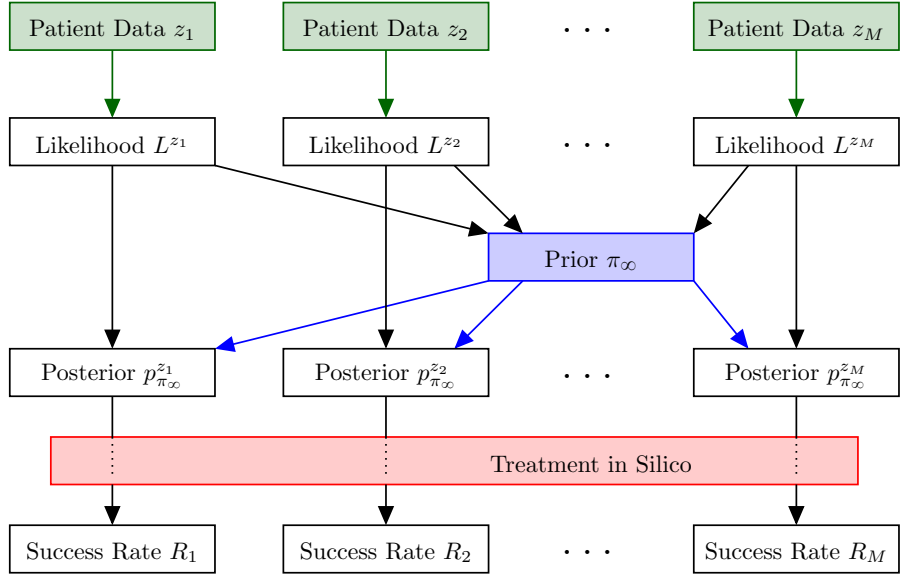
\begin{figure}%[H]
\begin{center}
\begin{tikzpicture}[scale=0.9, transform shape]
\pgfsetlinewidth{0.15ex}

% Prior
\fill[blue!20] (0.5,0.7) rectangle (3.5,1.4);
\draw[blue] (0.5,0.7) rectangle (3.5,1.4);
\draw node at (2,1) {\textcolor{black}{Prior $\pi_\infty$}};

\draw[blue,arrows = {-triangle 45}] (0.5,0.7) -- (-3.5,-0.3);
\draw[blue,arrows = {-triangle 45}] (1,0.7) -- (0,-0.3);
\draw[blue,arrows = {-triangle 45}] (3.5,0.7) -- (4.3,-0.3);

% Patients
\fill[dunkelgruen!20] (-6.5,4.1) rectangle (-3.5,4.8);
\draw[dunkelgruen] (-6.5,4.1) rectangle (-3.5,4.8);
\draw node at (-5,4.4) {\textcolor{black}{Patient Data $z_1$}};

\fill[dunkelgruen!20] (-2.5,4.1) rectangle (0.5,4.8);
\draw[dunkelgruen] (-2.5,4.1) rectangle (0.5,4.8);
\draw node at (-1,4.4) {\textcolor{black}{Patient Data $z_2$}};

\draw node at (2,4.4) {\textcolor{black}{\huge\ldots}};

\fill[dunkelgruen!20] (6.5,4.1) rectangle (3.5,4.8);
\draw[dunkelgruen] (6.5,4.1) rectangle (3.5,4.8);
\draw node at (5,4.4) {\textcolor{black}{Patient Data $z_M$}};

\draw[dunkelgruen,arrows = {-triangle 45}] (-5,4.1) -- (-5,3.1);
\draw[dunkelgruen,arrows = {-triangle 45}] (-1,4.1) -- (-1,3.1);
\draw[dunkelgruen,arrows = {-triangle 45}] (5,4.1) -- (5,3.1);

% Likelihood
\draw (-6.5,2.4) rectangle (-3.5,3.1);
\draw node at (-5,2.7) {Likelihood $L^{z_1}$};

\draw (-2.5,2.4) rectangle (0.5,3.1);
\draw node at (-1,2.7) {Likelihood $L^{z_2}$};

\draw node at (2,2.7) {\huge\ldots};

\draw (6.5,2.4) rectangle (3.5,3.1);
\draw node at (5,2.7) {Likelihood $L^{z_M}$};

\draw[arrows = {-triangle 45}] (-5,2.4) -- (-5,-0.3);
\draw[arrows = {-triangle 45}] (-1,2.4) -- (-1,-0.3);
\draw[arrows = {-triangle 45}] (5,2.4) -- (5,-0.3);

\draw[arrows = {-triangle 45}] (-3.5,2.4) -- (0.5,1.4);
\draw[arrows = {-triangle 45}] (0,2.4) -- (1,1.4);
\draw[arrows = {-triangle 45}] (4.3,2.4) -- (3.5,1.4);

% Posterior
\draw (-6.5,-1) rectangle (-3.5,-0.3);
\draw node at (-5,-0.7) {Posterior $p_{\pi_\infty}^{z_1}$};

\draw (-2.5,-1) rectangle (0.5,-0.3);
\draw node at (-1,-0.7) {Posterior $p_{\pi_\infty}^{z_2}$};

\draw node at (2,-0.7) {\huge\ldots};

\draw (6.5,-1) rectangle (3.5,-0.3);
\draw node at (5,-0.7) {Posterior $p_{\pi_\infty}^{z_M}$};

% Treatment
\fill[red!20] (-5.9,-2.3) rectangle (5.9,-1.6);
\draw[red] (-5.9,-2.3) rectangle (5.9,-1.6);
\draw node at (2,-2) {Treatment in Silico};

\draw (-5,-1) -- (-5,-1.6);
\draw[dotted] (-5,-1.6) -- (-5,-2.2);
\draw[arrows = {-triangle 45}] (-5,-2.3) -- (-5,-2.9);
\draw (-1,-1) -- (-1,-1.6);
\draw[dotted] (-1,-1.6) -- (-1,-2.2);
\draw[arrows = {-triangle 45}] (-1,-2.3) -- (-1,-2.9);
\draw (5,-1) -- (5,-1.6);
\draw[dotted] (5,-1.6) -- (5,-2.2);
\draw[arrows = {-triangle 45}] (5,-2.3) -- (5,-2.9);

% Succes Rate
\draw (-6.5,-3.6) rectangle (-3.5,-2.9);
\draw node at (-5,-3.3) {Success Rate $R_1$};

\draw (-2.5,-3.6) rectangle (0.5,-2.9);
\draw node at (-1,-3.3) {Success Rate $R_2$};

\draw node at (2,-3.3) {\huge\ldots};

\draw (6.5,-3.6) rectangle (3.5,-2.9);
\draw node at (5,-3.3) {Success Rate $R_M$};

\end{tikzpicture}

\caption{Algorithmic scheme for the computation of patient-specific parametrizations and predictions of
individual treatment success rates.}
\end{center}
\end{figure}

\textbf{Approach for scenario (B):}
\\
If the patients are diagnosed with diseases $s_1,\dots,s_L$ and the number of
patients $M_l$ is large for each disease $s_l$, then this extra information
can be used by applying the procedure described in (A) to each subset
$\mathcal Z^{(l)}$ separately in order to obtain more precise results.

%% file: ToyExample.tex
\newpage

\section{Toy Example}
\label{section:ToyExample}
We will start with an easy to grasp low-dimensional non-identifiable mechanical
example for scenario (A), where the patients will be
represented by springs with different stiffness values as possible parameter values. Different stiffness values result in
different system responses when a certain force is applied to the springs,
which represents a treatment of the patients. The example demonstrates how our
algorithm can be applied to predict success rates of such a treatment.

We buy two boxes (representing two diseases) of springs, the
first containing 150 springs with stiffness $K_1 = 15\, N/m$, the
second containing 150 springs with stiffness $K_2 = 30\, N/m$.
The springs are of a low quality and their actual stiffness varies from the
nominal value with a standard deviation of $15\%$ (we assume a normal
distribution for each box).

Once we arrive at home, we realize that the boxes are not labeled and that we
already forgot the values $K_1$ and $K_2$ as well as the standard deviations.
As described above, there are two possible scenarios and we will only treat
scenario (A), since (B) goes analogously:
\begin{enumerate}[(A)]
  \item We mix up the springs by putting all of them into one big box (no
  diagnosis for each spring).
  \item We keep them in the two separate boxes (the springs are diagnosed with
  diseases $s_1$ or $s_2$, depending on the box they come from).
\end{enumerate} 
In order to determine the stiffness of a single spring, we perform the following
experiment (harmonic oscillator, see Figure \ref{fig:ToyExampleArrangement}):
\begin{itemize}
  \item We fix one end of the spring and put a mass $m=700\mathrm{g}$ to the
  other end.
  \item After compressing it by 10cm, we let it swing. Applying Hooke's law this results in
  the following ODE
  \begin{equation}
  \label{equ:federODE}
  x''(t) = -\frac{K}{m}\, x(t),\qquad x(0)=-10\mathrm{cm}.
  \end{equation}
  \item  We measure its amplitude at time $t^\ast = 1$s. Therefore, the model
  $\phi:\R\to\R$ and the measurement $Z\in\R$ are given by
  \[
  \phi(K) = x(t^\ast),\qquad Z = \phi(K) + E,
  \] 
  where the measurement error $E\sim\rho_E = \cN(0,\sigma_E^2)$ is assumed to be
  standard normal distributed with mean $0$ and a standard deviation of
  $\sigma_E = 1$cm.    
\end{itemize}

\begin{figure} %[H]
\begin{center}
\begin{tikzpicture}%[width=\textwidth]
[scale=0.8, transform shape]
% \pgfsetlinewidth{0.15ex}
\node[circle,fill=blue,inner sep=2.2mm,opacity=0.2] (c) at (5.2,0) {};
\draw[decoration={aspect=0.3, segment length=13, amplitude=6,coil},decorate,color=black!20] (-8,0) -- (c);

\node[circle,fill=blue!40,inner sep=2.2mm] (a) at (0,0) {};
\draw[decoration={aspect=0.3, segment length=6, amplitude=6,coil},decorate,opacity=0.4] (-8,0) -- (a);

\node[circle,fill=blue,inner sep=2.2mm] (b) at (-4,0) {};
\draw[decoration={aspect=0.3, segment length=3, amplitude=6,coil},decorate] (-8,0) -- (b);

\fill [pattern = north east lines] (-8,-0.55) rectangle (6.5,-0.35);
\draw[thick] (6.5,-0.35) -- (-8,-0.35) -- (-8,0.7);
\fill [pattern = north east lines] (-8,-0.55) rectangle (-8.2,0.7);
% \draw[thick] (-8,-0.35) -- (-8,0.7);

\node at (-3.44,0.9) {\small{$x_0=-10$}};
\draw[dashed] (-4,-0.35) -- (-4,0.75);
\node at (0.56,0.9) {\small{$x_{\mathrm{rest}}=0$}};
\draw[dashed] (0,-0.35) -- (0,0.75);
\node at (5.76,0.9) {\small{$x_{\mathrm{Bell}}=13$}};
\draw[dashed] (5.2,-0.35) -- (5.2,0.75);

\draw (6,-0.1) ellipse (0.5 and 0.15);
\draw[fill=black] (5.8,-0.06) circle (0.1);
\draw[rounded corners=4pt] (5.5,-0.1) -- (5.7,0.2) -- (5.7,0.5) -- (6,0.7) -- (6.3,0.5) -- (6.3,0.2) -- (6.5,-0.1);
\end{tikzpicture}
\caption{Experimental arrangement for the toy example described above.}
\label{fig:ToyExampleArrangement}
\end{center}
\end{figure}
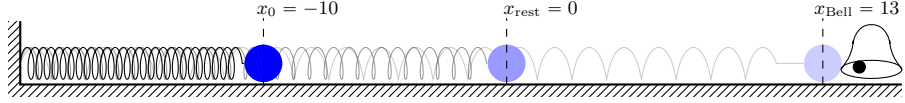

We implemented the fixed point iterations for NPMLE, DS-MLE and MPLE discussed
in Section \ref{section:Theory} (for the latter $\gamma=49$ appeared adequate), starting with a ``non-informative'' prior
$\pi_0$, which we chose as the uniform distribution on
$[1,50]$ (in $kg/s^2$), as well as the computation of the corresponding
posteriors for one of the springs.
The results are shown in Figure \ref{fig:fixedPointIterationToyPriors}.
Since the model is highly non-identifiable (for each point
$x^\ast\in[-10,10]$ there are several stiffness values $K$ for which the
trajectory of \eqref{equ:federODE} goes through $(t^\ast,x^\ast)$, i.e.
$\phi(K) = x^\ast$), the prior estimates are rather non-informative. Therefore,
the effect on the posterior is barely visible, while the smoothing of the
prior is evident.

\begin{figure}%[H]
\centering
\includegraphics[width=1\textwidth]{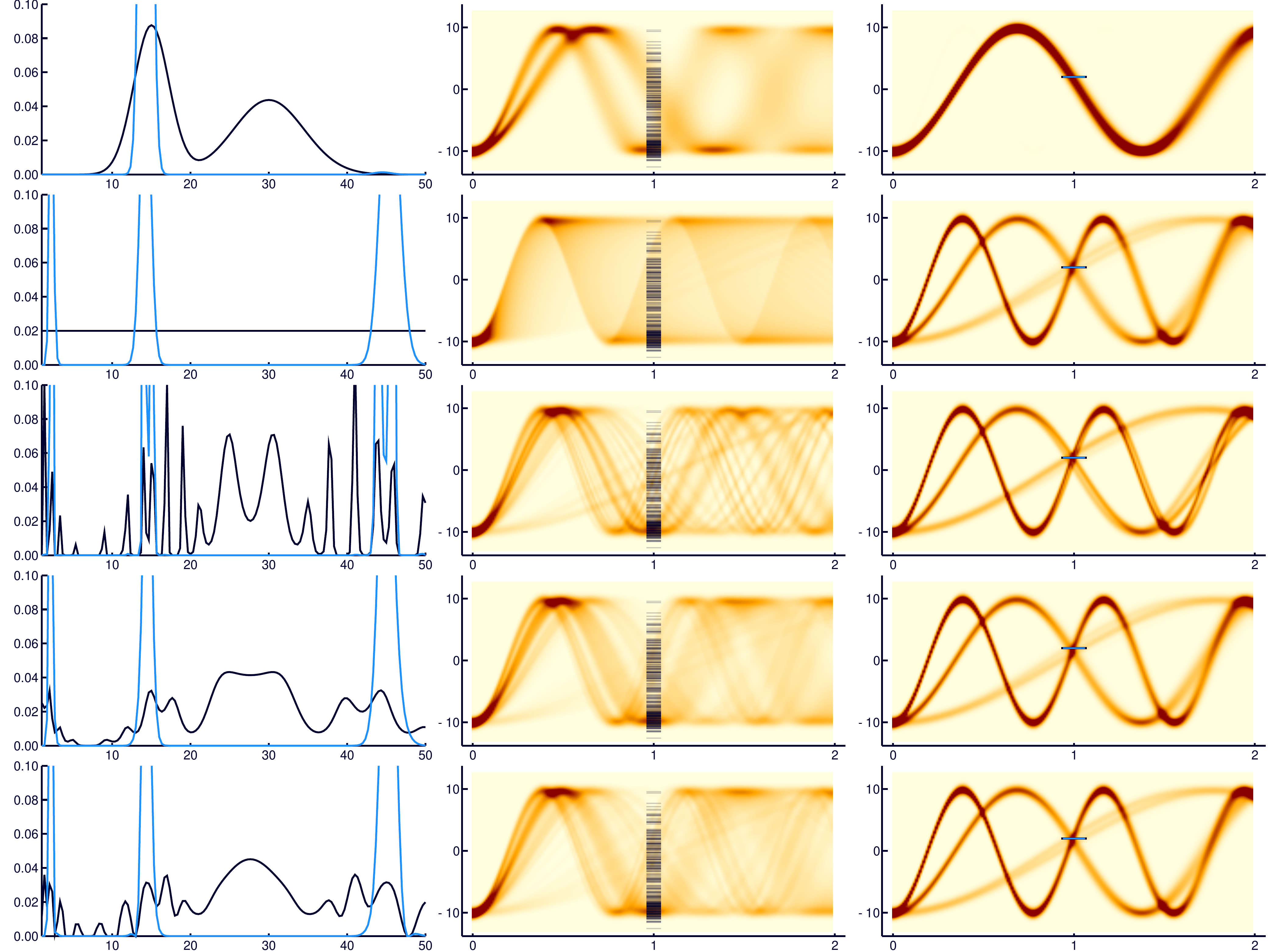}
\caption{From top to bottom: priors $\pi_\true$, $\pi_0$,
$\pi_\NPMLE$, $\pi_\DSMLE$, $\pi_\MPLE$.
\\
From left to right: the prior (black) and the corresponding posterior for one specific spring (blue), distribution of trajectories sampled from that prior plotted over time and all
measurements, measurement of that spring and corresponding posterior distribution of trajectories.
\\
500 iterations have been performed. Note that for
a higher number of iterations $\pi_\NPMLE$ and $\pi_\DSMLE$ (due to the
discretization method discussed in Section \ref{section:TheoryNumerical}) will
not stop peaking.
}
\label{fig:fixedPointIterationToyPriors}
\end{figure}

The ``treatment procedure'' will be modeled by hitting the mass in positive
$x$-direction with several pulses at certain times given by the force $F(t)$ plotted in Figure
\ref{fig:Force},
which results in the following perturbed ODE:
\[
x''(t) = -\frac{K}{m}x(t) + F(t),\qquad x(0)=-10\mathrm{cm}.
\]
\begin{figure}%[H]
        \centering
        \includegraphics[width=1\textwidth]{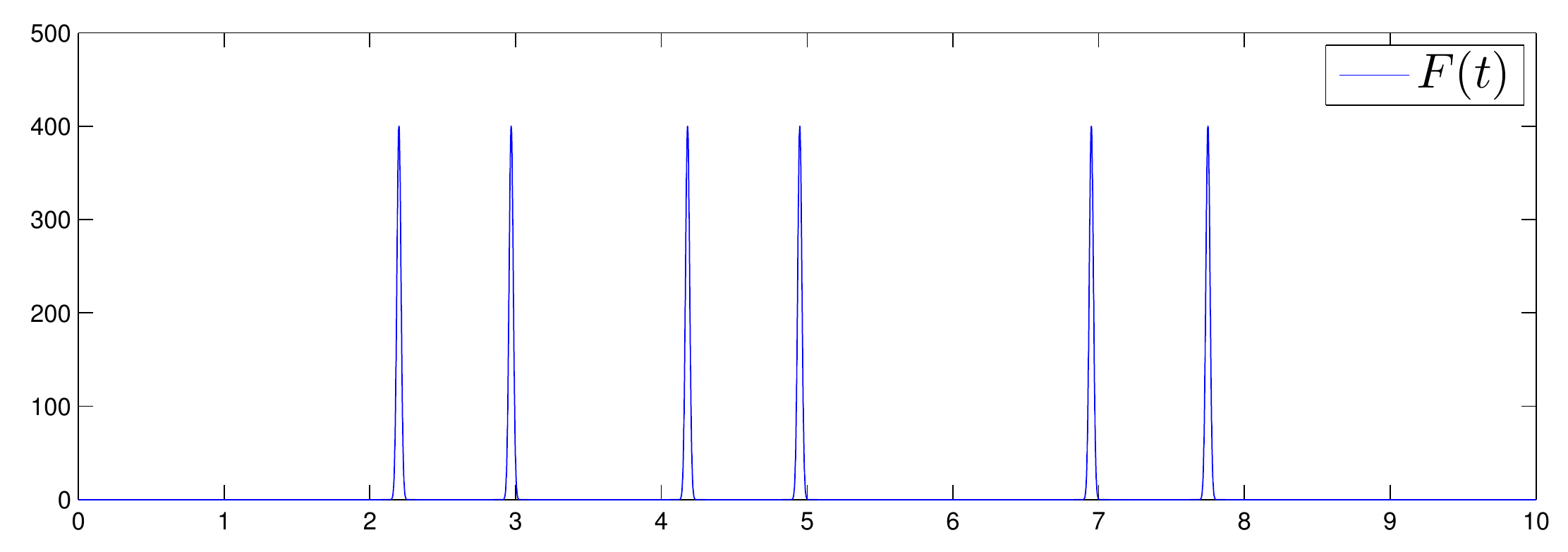}
        \caption{The force $F$ used for modeling the treatment procedure
        plotted over time $t$.}
	\label{fig:Force}
\end{figure}
The treatment will be considered successful, if the mass hits the bell located
at $x_{\mathrm{Bell}} = 13$cm within ten seconds.

For each $\Xi\in\{\true,0,\NPMLE,\DSMLE,\MPLE\}$, we computed the success rates
$R_m^\Xi$, $m=1,\dots,M$ with respect to $\pi_\Xi$ and their (empirical)
standard deviations from the true success rates:
\[
\sigma^\Xi_R = \sqrt{\frac{1}{M-1}\sum_{m=1}^M \left|R_m^\Xi -
R_m^\true\right|^2}.
\]
The results are given in Table \ref{table:successrateSTD}.

\begin{table}[H]
\centering
\begin{tabular}{|C{1.5cm}||C{1.5cm}|C{1.5cm}|C{1.5cm}|C{1.5cm}|}
\hline
\vspace{0.1cm}$\Xi$ \vspace{0.1cm} & $0$ & $\NPMLE$ & $\DSMLE$ & $\MPLE$
\\
\hline
\vspace{0.1cm}$\sigma^\Xi_R$\vspace{0.1cm} & 0.126 & 0.135 & 0.106 & 0.097\\
\hline
\end{tabular}
\caption{The standard deviations of the predicted success rates for different
prior distributions $\pi_\Xi$.}
\label{table:successrateSTD}
\end{table}

\subsection{Example of an unsmooth reference prior}
\label{section:goWrong}
By the following example, we want to make the reader aware of the fact that regularizing by means of the mutual information $\cI[X;Z]$ not always results in a smooth prior. Even if we ignore the likelihood in \eqref{equ:MPLEformulation} and maximize $\cI[X;Z]$ alone, the resulting prior, which is the reference prior 
\[
\pi_{\rm ref} = \argmax_\pi\ \cI[X;Z](\pi),
\]
can be very irregular.

If, in the above example, we perform two measurements at times $t_1^\ast = 1$s
and $t_2^\ast = 1.7$s instead of only one, the model $\phi:\R\to\R^2$ and the
measurement $Z\in\R^2$ are given by
\begin{align*}
&\phi(K) = (x(t_1),x(t_2))^{\intercal},\qquad Z = \phi(K) + E,
\intertext{with error}
&E = (E_1,E_2)^{\intercal}\, \sim \, \rho_E =
\cN(0,\sigma_E^2)\otimes \cN(0,\sigma_E^2).
\end{align*}
The resulting reference prior shown in Figure \ref{fig:goWrong} is very unsmooth.
Let us analyze the reason for this!
%
%In this case, penalizing by means of the mutual information $\cI[X;Z]$ does not
%yield any smoothing. On the contrary, if we look at the density that maximizes
%$\cI[X;Z]$ alone (without the likelihood), which is the reference prior
%\[
%\pi_{\rm ref} = \argmax_\pi\ \cI[X;Z](\pi),
%\]
%the result is very irregular, as shown in Figure \ref{fig:goWrong}.
\begin{figure}[H]
        \centering
        \includegraphics[width=1\textwidth]{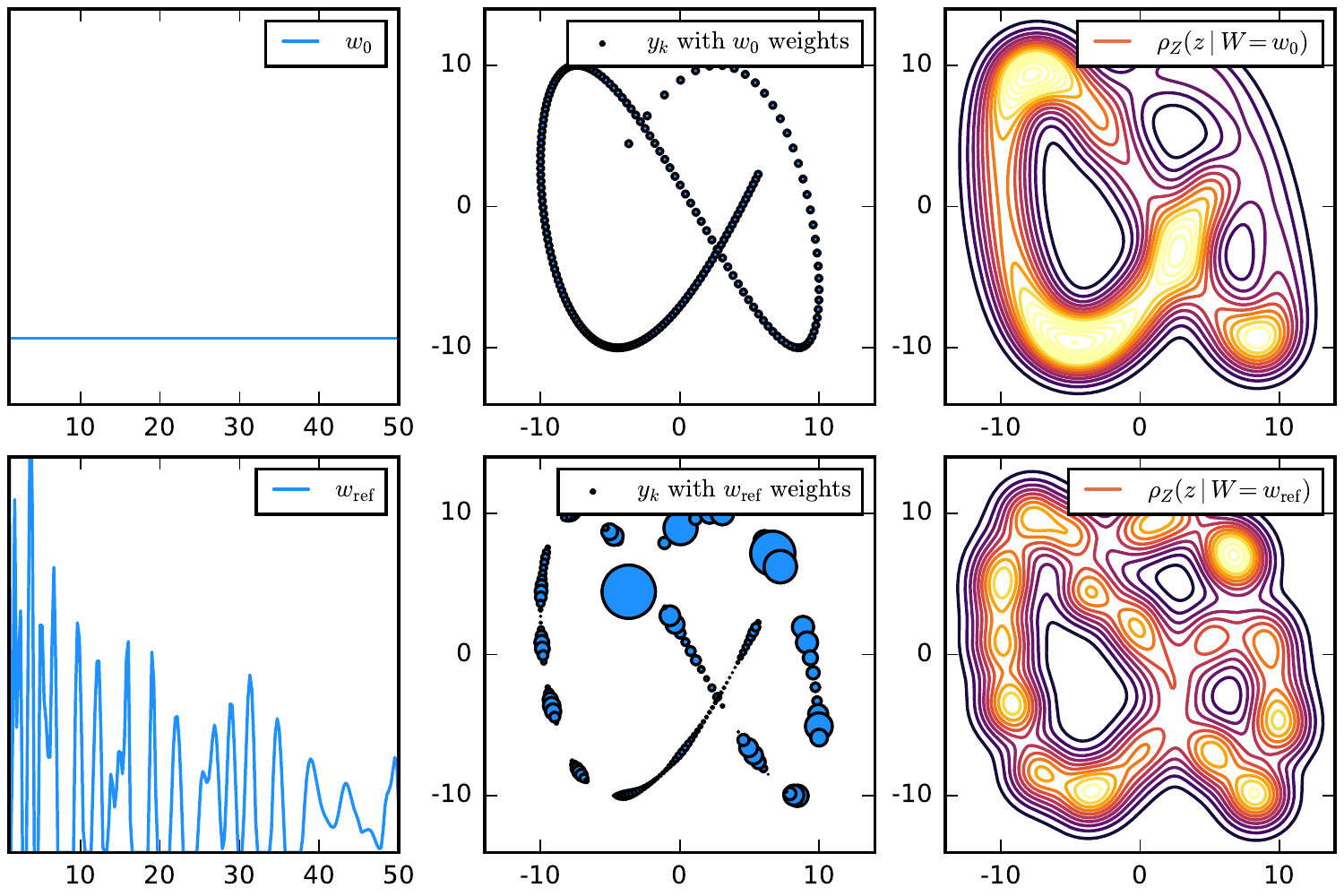}
        \caption{
        From left to right: Two different priors (uniform and
        reference prior), resulting weights in the linear combination
        \eqref{equ:rho_ZLinearCombination}, resulting densities $\rho_Z(z\mid
        W=w)$ in measurement space. One can clearly see that the reference prior
        results in a higher entropy $H_Z(\pi)$ and thereby in a higher mutual
        information $\cI[X;Z](\pi)$, which it was optimized for. However, it is
        very irregular.
        }
	\label{fig:goWrong}
\end{figure}
% This seems to be an example where both, the reference prior and the penalization
% by mutual information, seem to fail.

If we discretize the parameter space $\cX = [1,50]$ by 200 equidistant points $x_k$ and define $y_k := \phi(x_k)$, the density $\rho_Z(z\mid W=w)$ consists
of a linear combination of shifted versions of $\rho_E$,
\begin{equation}
\label{equ:rho_ZLinearCombination}
\rho_Z(z\mid W=w) = \sum_{k=1}^K w_k\, \rho_E(z-y_k),
\end{equation}
in our case a Gaussian mixture with centers $y_k$.
The aim is to choose $w = (w_k)_k$ in such a way that it maximizes the entropy
$H_Z(w)$, i.e. to make $\rho_Z(z\mid W=w)$ rather ``spread'' and ``flat''. 

Therefore, we get huge weights $w_k$ for those $k$, where $\rho_E(z-y_k)$ covers
a large region that cannot be covered by any other $\rho_E(z-y_j)$.
Low weights $w_k$ are chosen small if the corresponding $y_k$ lie in a regions
that are already covered by many $\rho_E(z-y_j)$ (otherwise the density
$\rho_Z(z\mid W=w)$ would get too high around $y_k$). This way, we arrive at
the peaked prior $\pi_{\rm ref}$ shown in Figure \ref{fig:goWrong}.

% If we plot the ``path'' of $\phi$ (its image in $\R^2$) as in Figure
% TODO, it is evident that putting a huge weight $w_i$ for $y_i$ will give a
% rather high entropy, since $\rho_E(z-y_i)$ covers a large region that cannot be
% covered by any other $\rho_E(z-y_k)$. If the weight of $w_i$ is large, then its
% neighbours $w_{i-1}$ and $w_{i+1}$ have to be rather low, since we do not want
% the density to get too high around $y_i$.
% This results in the peaked behaviour of $\pi_{\rm ref}$ in Figure TODO.

The unsmoothness of the reference prior $\pi_{\rm ref}$ (and of our MPLE) is not necessarily a downside, but it is worthwhile to be aware of such situations when using these methods.

%The issue discussed here seems to be a general problem if the model
%$\phi\colon\R^d\to\R^n$ maps (non-trivially) into a higher dimensional space
%than the parameter space ($n>d$), as will be the case in our next example.

%% file: LargeExample.tex
\newpage
\section{Parameter estimation in a large ODE model}
\label{section:LargeExample}
As model system, we consider a model for the human menstrual cycle, named
GynCycle \cite{Roeblitz2013}. This model is given by a system of 33 ordinary
differential equations (ODEs) and 114 parameters.
It has been calibrated previously with
time-series data of blood concentrations for four hormones from 12 patients
during the unperturbed cycle and during treatment (dose-response experiments).
Using deterministic, local optimization (an error-oriented Gauss-Newton
method), only 63 out of 114 parameters could be identified from the given data.
The remaining parameters kept their values from previous versions of the model.
In the following, we will denote these parameter values as nominal values. The
model is currently used to make patient-specific predictions about the outcome
of treatment strategies in reproductive medicine\footnote{EU project
PAEON-Model Driven Computation of Treatments for Infertility Related
Endocrinological Diseases, project number 600773.}. Hence, quantification of
uncertainty in these predictions is of utmost importance.

We got access to additional measurement values of 36 woman for the four
hormones LH, FSH, E2 and P4 during normal cycles\footnote{Courtesy of Dorothea
Wunder, University Hospital of Lausanne.}. These data are sparse or incomplete
in the sense that measurements were not taken on all cycle days, resulting in
about 15 measurement time points per patient and hormone. Our approach,
however, is flexible enough to handle such a data situation.

Based on these data, our aim is to estimate the prior distribution for 82 out
of the 114 parameters $\Theta$ (Hill exponents have been fixed), denoted by
$\pi_0^\Theta$, as well as for the initial values $Y_0$ of the 33 model
species, $\pi_0^{Y_0}$, resulting in a total of 115 dimensions: $X =
(Y_0,\Theta)\in\R^{115}$.
%, where we assume that $Y_0$ and $\Theta$ are independent

As an initial guess for the prior $\pi_0^\Theta$ of model parameters, we have
chosen uniform distributions on the intervals between zero and five times the
nominal parameter values as non-informative approach.
For $\pi_0^{Y_0}$, we chose a mixture distribution of independent normals around
daily values of a reference cycle computed with the nominal parameters, i.e.
\[
\pi_0^{Y_0} = \frac{1}{31} \sum_{t=0}^{30} G\left[y_{\mathrm{ref}}(t), \Sigma \right],
\]
where $G[m,C]$ denotes the Gaussian density with mean $m$ and covariance
matrix $C$, $y_{\mathrm{ref}} : \mathbb{R}\rightarrow\mathbb{R}^{33}$ is the
reference solution over one menstrual cycle and $\Sigma$ is a diagonal matrix
consisting of the squared standard deviations of each species, respectively,
\[
\Sigma = \mathrm{diag}(\sigma_1^1,\dots,\sigma_{33}^2),
\qquad
\sigma_j^2 = \frac{1}{30}\sum_{t=0}^{30} |y_{\mathrm{ref},j}(t) - \overline y_{\mathrm{ref},j}|^2.
\]
The total prior $\pi_0\colon\R^{115}\to\R$ is build up from $\pi_0^{Y_0}$ and
$\pi_0^{\Theta}$ under the assumption that $Y_0$ and $\Theta$ are independent,
\[
\pi_0(y_0,\theta) := \pi_0^{Y_0}(y_0)\, \pi_0^\Theta(\theta).
\]

The likelihood for specific measurements $z\in\R^{4\times 31}$ is chosen
normally distributed with a (relative) standard deviation of $\sigma = 10\%$,
\[
\rho_Z(z\, |\, X=x)\, \propto\, \exp \left(-\frac{d(\phi(x),z)^2}
{2 \sigma^2} \right),
\]
where
\[
\ 
d(u,v)^2 = \sum_{j=1}^4 \sum_{t=0}^{30} \left|\frac
{u_j(t) - v_j(t)}{c_j}\right|^2
\]
is the relative distance between simulated and measured data, $c_j$ are suitably
chosen constants of the magnitude of the measurements
$\left(v_j(t)\right)_{t=0,\dots,30}$,
\[
\phi(x) = \phi(y_0,\theta) = \left(\mathrm{proj}_4(y(t))\right)_{t=0,\dots,30},
\]
$\mathrm{proj}_4$ denotes the projection onto the four measurable components,
and $y(t)$ is the solution of the GynCycle model with initial values $y_0$ and parameters $\theta$.
\begin{remark}
As mentioned above, the measurements for most women were not taken daily, resulting in incomplete data. In this case, $\phi$ and $d$ have to be chosen separately for each woman, restricting them to measured components. This does not influence the applicability of our algorithm.
\end{remark}

\begin{figure} %[H]
\centering
\includegraphics[width=1\textwidth]{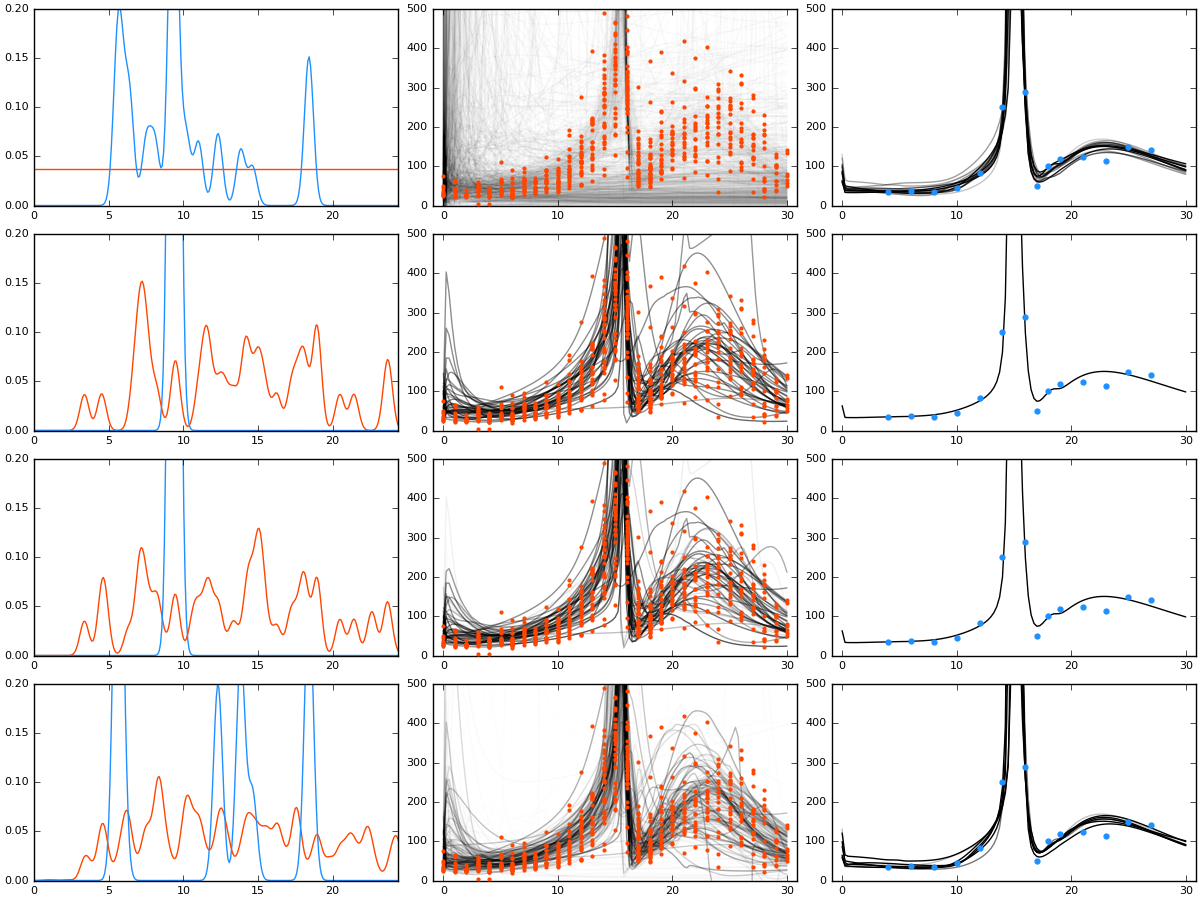}
\caption{From top to bottom:
prior given by $\pi_0$,
$\pi_\NPMLE$, $\pi_\DSMLE$, $\pi_\MPLE$.
\\
From left to right: marginal distribution of the prior (black) for one parameter (the transition rate from the late primary follicle PrA2 to the early secondary follicle SeF1) and the corresponding posterior for one specific patient (blue), trajectories sampled from that prior and all measurements, measurements of that patient and trajectories sampled from the posterior.
\\
300 iterations have been performed. Note that for
a higher number of iterations $\pi_\NPMLE$ and $\pi_\DSMLE$ (due to the
discretization method discussed in Section \ref{section:TheoryNumerical}) will
not stop peaking.}
\label{fig:largeODEexample}
\end{figure}

We sampled the posteriors for each patient (in order to get $\pi_1$-distributed
samples as explained in Section \ref{section:TheoryNumerical}) using the
adaptive mixture Metropolis algorithm by Roberts and Rosenthal
\cite{roberts2009examples}, which is basically a multivariate
Metropolis-Hastings algorithm tuning its Gaussian proposal density for the
current sample based on the covariance of the former ones.
As the computation is independent for each patient, this problem is
well-scalable in the number of patients and we were thus able to compute 10 million samples
for each patient. The Raftery-Lewis diagnostic suggests around 7 million samples
for convergence and the Gelman and Rubin criterion confirms this in our case
with potential scale reduction factors smaller than 1.05.

Once the $\pi_1$-distributed samples had been computed, we implemented the fixed
point iterations for NPMLE, DS-MLE and MPLE discussed
in Section \ref{section:Theory} (for the latter $\gamma=19$ appeared adequate), as well as the corresponding posteriors for one
of the springs. The results are shown in Figure \ref{fig:largeODEexample}.
From the plots it seems counterintuitive that the posterior $p_{\pi_0}^{z_m}$
stemming from the prior $\pi_0$ is so much less informative then those coming
from the estimated priors. However, this is less surprising if one keeps in
mind that we only consider \emph{marginal} densities and therefore the priors
$\pi_\NPMLE$, $\pi_\DSMLE$ and $\pi_\MPLE$ can be much more informative
(compared to $\pi_0$) than they look for each marginalization (e.g., they might
be concentrated around a submanifold of $\cX$). This demonstrates the
importance and strength of empirical Bayes methods.

Our next step will be to compute individual success rates for treatments
frequently used in reproductive medicine (modeled by a perturbed ODE)
and to compare the results with clinical outcomes.
%
%\begin{figure}[H]
%\centering
%\includegraphics[width=1\textwidth]{images/gyncplots}
%\caption{From top to bottom:
%prior given by $\pi_0$,
%$\pi_\NPMLE$, $\pi_\DSMLE$, $\pi_\MPLE$.
%\\
%From left to right: marginal distribution of the prior (with iterations in
%lighter blue) for one parameter (the transition rate from the late primary
%follicle PrA2 to the early secondary follicle SeF1), distribution of
%trajectories sampled from that prior plotted over time together with all
%measurements, posterior distribution of trajectories with respect to that prior
%for one specific patient.
%\\
%300 iterations have been performed. Note that for
%a higher number of iterations $\pi_\NPMLE$ and $\pi_\DSMLE$ (due to the
%discretization method discussed in Section \ref{section:TheoryNumerical}) will
%not stop peaking.}
%\label{fig:largeODEexample}
%\end{figure}

%% file: Conclusion.tex
\section{Conclusion}
\label{section:Conclusion}
We have introduced a method that estimates the prior distribution in the
empirical Bayes framework, when measurements for a large number of individuals are available.
We discussed the issue of convergence in the identifiable case and also what
happens in the non-identifiable case.

A detailed scheme for the numerical realization of the method has been
elaborated, see Sections \ref{section:TheoryNumerical} and
\ref{section:ResultingAlgorithm}. The numerical
approximation to the fixed point iteration has to be applied with caution,
since its convergence properties differ from the ones of the exact iteration,
see the discussion in the introduction and in Section
\ref{section:TheoryNumerical}. A transformation-invariant regularization
approach has been applied in order to deal with this situation.

The method has been applied to a toy example in low dimensions to confirm our
theoretical results, as well as to a high-dimensional real life problem.
As a byproduct, the method can be applied to deconvolve blurred images, as
discussed in Appendix \ref{section:Deconvolution}.

As demonstrated and explained in Section \ref{section:goWrong}, the resulting $\pi_\MPLE$ can be rather unsmooth. As stated by Good and Gaskins in \cite{good1971nonparametric}, ``continuous distributions [\ldots] could have violent small ripples with little effect on the entropy''. They advise against the use of entropy as a roughness penalty in the continuous case. From our point of view, ``small ripples'' in the density are not a criterion for the
exclusion of a distribution.
Not only can the true distribution itself have ripples, but, even if not so, a
distribution with small ripples can be a good approximation to it.
However, if the aim is to get a smooth prior, this approach might be inadequate.
It is worth mentioning that in this case the reference prior encounters the same
problem. A possible solution is given by DS-MLE, which was also discussed and
implemented for both problems.

We wish to extend the application in Section \ref{section:LargeExample} to data
bases of thousands of patients (instead of just 36), which are available by now
at the University Hospitals in Zurich and Basel.

%% file: Deconvolution.tex
\section{Deconvolution of Blurred Images}
\label{section:Deconvolution}

The fixed point iteration
\eqref{equ:fixedPointIteration}--\eqref{equ:IterationPsi} can readily be applied
for the deconvolution of blurred images with known point spread function
(non-blind deconvolution).

For this, the probability densities $\pi_n$ are viewed as blurred versions of
the true density $\pi_{\rm true} = \rho_X$, whereby $\rho_X$ is smoothed by a
convolution kernel $G_n$:
\[
\pi_{n+1}(x)
=
\Psi\pi_n(x)
=
\int p_{\pi_n}^z(x)\, \rho_Z(z)\, \mathrm dz
=
\int \rho_X(\tilde x)
\underbrace{\int p_{\pi_n}^z(x)\, \rho_Z(z\, |\, X=\tilde x)\, \mathrm dz}_{=: G_{n+1}(x,\tilde x)}
\, \mathrm d\tilde x.
\]
With growing number $n$ the iterates become less smoothed, converging
to $\rho_X$. Therefore, the fixed point iteration
\eqref{equ:fixedPointIteration}--\eqref{equ:IterationPsi} results in a deconvolution process of $\pi_0$ to the original prior $\rho_X$.

\begin{figure} %[H]
        \centering
        \begin{subfigure}[b]{0.32\textwidth}
                \centering
                \includegraphics[width=1\textwidth]{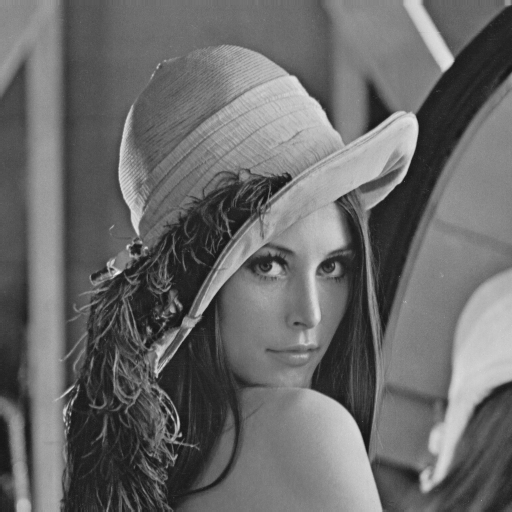}
                \caption{original image}
	\end{subfigure}%
        \hfill
        \begin{subfigure}[b]{0.32\textwidth}
                \centering
                \includegraphics[width=1\textwidth]{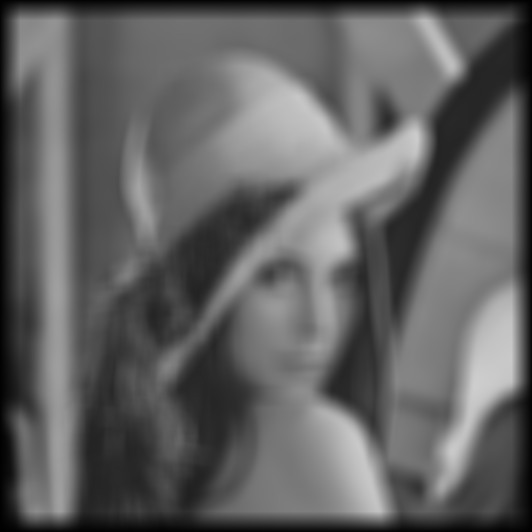}
                \caption{blurred image}
        \end{subfigure}
        \hfill
        \begin{subfigure}[b]{0.32\textwidth}
                \centering
                \includegraphics[width=1\textwidth]{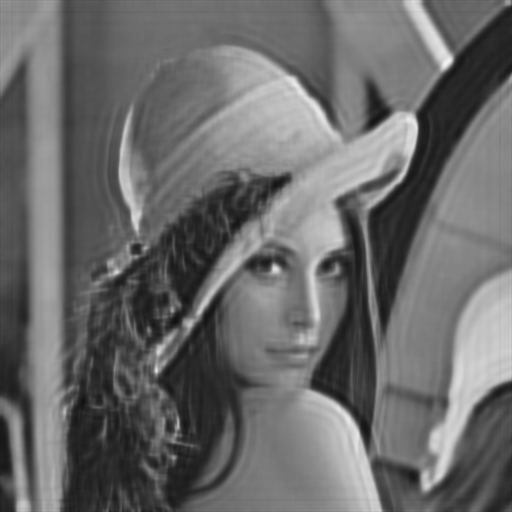}
                \caption{deconvolved, 50 iter.}
        \end{subfigure}%
        \vspace{0.2cm}
        \begin{subfigure}[b]{0.32\textwidth}
                \centering
                \includegraphics[width=1\textwidth]{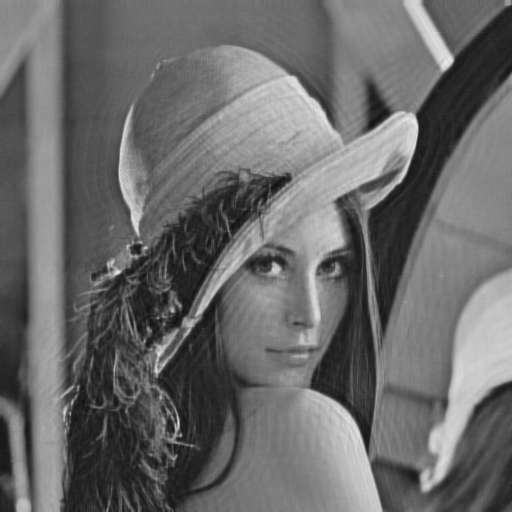}
                \caption{deconvolved, 1500 iter.}
        \end{subfigure}        
        \hfill
        \begin{subfigure}[b]{0.32\textwidth}
                \centering
                \includegraphics[width=1\textwidth]{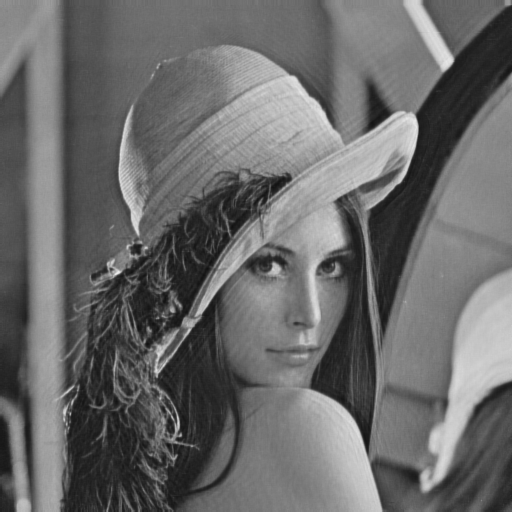}
                \caption{deconvolved, 10000 iter.}
	\end{subfigure}%
        \hfill
        \begin{subfigure}[b]{0.32\textwidth}
                \centering
                \includegraphics[width=1\textwidth]{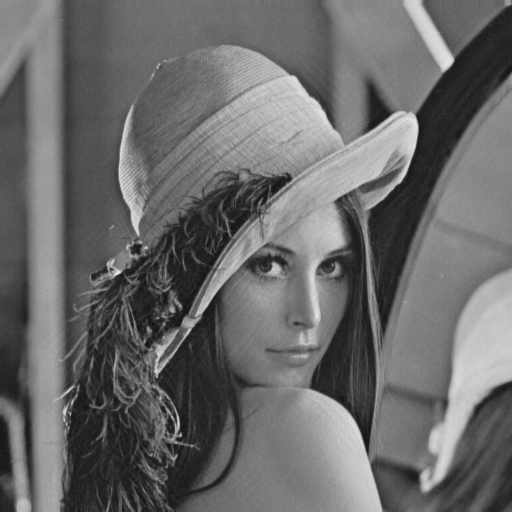}
                \caption{deconvolved, 50000 iter.}
        \end{subfigure}                
        \caption{Deconvolution of an artificially blurred image (b) using the fixed point iteration
        \eqref{equ:fixedPointIteration}--\eqref{equ:IterationPsi}.}
        \label{fig:Deconvolution}
\end{figure}

In fact, if we choose $\phi=\mathrm{Id}$ and the error density $\rho_E$ as
the point spread function, we can view $\rho_X:\R^2\to\R_{\ge 0}$  as the
original image (without loss of generality, we can assume that it is normalized
and given by gray scale values) and the evidence
\[
\rho_Z(z) = \int \rho_X(x)\, \rho_E(z-\phi(x))\, \mathrm dx = (\rho_X\ast\rho_E)(z)
\]
as the blurred image.
In this setup, our algorithm provides a method for the restoration of the original image from the blurred image,
as demonstrated in Figure \ref{fig:Deconvolution}.